\documentclass[runningheads,a4paper]{llncs}

\usepackage{algorithm}
\usepackage[noend]{algpseudocode}
\usepackage{amsmath}
\usepackage{graphicx}
\usepackage{tabularx}

\newcommand{\Geqt}{\ensuremath{G_\text{eqt}}}
\newcommand{\Veqt}{\ensuremath{V_\text{eqt}}}
\newcommand{\Eeqt}{\ensuremath{E_\text{eqt}}}
\newcommand{\E}{\mbox{E}}
\newcommand{\OPT}{\ensuremath{\text{OPT}}}

\begin{document}

\mainmatter

\title{On the Runtime of Universal Coating for Programmable Matter}

\titlerunning{Runtime of Universal Coating}

\author{Joshua J. Daymude\inst{1}\footnote[1]{Supported in part by NSF grants CCF-1353089, CCF-1422603, and  REU--026935.} \and
    Zahra Derakhshandeh\inst{1}\footnotemark[1] \and
    Robert Gmyr\inst{2}\footnote[2]{Supported in part by DFG grant SCHE 1592/3-1.} \and
    Alexandra Porter\inst{1}\footnotemark[1] \and
    Andr\'ea W.\ Richa\inst{1}\footnotemark[1]\and
    Christian Scheideler\inst{2}\footnotemark[2] \and
    Thim Strothmann\inst{2}\footnotemark[2]}

\authorrunning{Daymude, et al.}

\institute{Computer Science, CIDSE, Arizona State University, USA, \\
    \{jdaymude,zderakhs,amporte6,aricha\}@asu.edu \\
    \and
    Department of Computer Science, Paderborn University, Germany, \\
    \{gmyr,scheidel,thim\}@mail.upb.de
}

\toctitle{Lecture Notes in Computer Science}
\tocauthor{Authors' Instructions}
\maketitle

\begin{abstract}
Imagine coating buildings and bridges with smart particles (also coined smart paint) that monitor structural integrity and sense and report on traffic and wind loads, leading to technology that could do such inspection jobs faster and cheaper and increase safety at the same time.
In this paper, we study the problem of uniformly coating objects of arbitrary shape in the context of {\it self-organizing programmable matter}, i.e., programmable matter which consists of simple computational elements called particles that can establish and release bonds and can actively move in a self-organized way. Particles are anonymous, have constant-size memory, and utilize only local interactions in order to coat an object.
We continue the study of our Universal Coating algorithm by focusing on its runtime analysis, showing that our algorithm terminates within a {\it linear number of rounds} with high probability. We also present a matching linear lower bound that holds with high probability. We use this lower bound to show a {\it linear lower bound on the competitive gap} between fully local coating algorithms and coating algorithms that rely on global information, which implies that our algorithm is also optimal in a competitive sense. Simulation results show that the competitive ratio of our algorithm may be better than linear in practice.
\end{abstract}

%
%

\vspace{-.3in}
\section{Introduction} \label{sec:intro}
\vspace{-.05in}
Inspection of bridges, tunnels, wind turbines, and other large civil engineering structures for defects is a time-consuming, costly, and potentially dangerous task. In the future, {\it smart coating} technology, or {\it smart paint}, could do the job more efficiently and without putting people in danger. The idea behind smart coating is to form a thin layer of a specific substance on an object which then makes it possible to measure a condition of the surface (such as temperature or cracks) at any location, without direct access to the location. The concept of smart coating already occurs in nature, such as proteins closing wounds, antibodies surrounding bacteria, or ants surrounding food to transport it to their nest. These diverse examples suggest a broad range of applications of smart coating technology in the future, including repairing cracks or monitoring tension on bridges, repairing space craft, fixing leaks in a nuclear reactor, or stopping internal bleeding.
We continue the study of coating problems in the context of self-organizing programmable matter consisting of simple computational elements, called particles, that can establish and release bonds and can actively move in a self-organized way using the geometric version of the amoebot model presented in~\cite{Daymude2016,Derakhshandeh2014}. In doing so, we proceed to investigate the runtime analysis of our Universal Coating algorithm, introduced in~\cite{Derakhshandeh2017}.
We first show that coating problems do not only have a (trivial) linear lower bound on the runtime, but that there is also a linear lower bound on the competitive gap between the runtime of fully local coating algorithms and coating algorithms that rely on global information.
We then investigate the worst-case time complexity of our Universal Coating algorithm and show that it terminates within a linear number of rounds with high probability (w.h.p.)\footnote{By {\it with high probability}, we mean with probability at least $1-1/n^c$, where $n$ is the number of particles in the system and $c>0$ is a constant.}, which implies that our algorithm is optimal in terms of worst-case runtime and also in a competitive sense.
Moreover, our simulation results show that in practice the competitive ratio of our algorithm is often better than linear.

\vspace{-.1in}
\subsection{Amoebot model} \label{sec:model}
\vspace{-.05in}
In the {\it amoebot model}, space is modeled as an infinite, undirected graph $G$ whose vertices are positions that can be occupied by at most one particle and whose edges represent all possible atomic transitions between these positions.
In the {\it geometric amoebot model}, we further assume that $G = \Geqt$, where $\Geqt = (\Veqt, \Eeqt)$ is the infinite regular triangular grid graph (see Figure~\ref{fig:model}a).
Each particle occupies either a single node (i.e., it is {\it contracted}) or a pair of adjacent nodes in $\Geqt$ (i.e., it is {\it expanded}), as in Figure~\ref{fig:model}b. Particles move by executing a series of {\it expansions} and {\it contractions}: a contracted particle can expand into an unoccupied adjacent node to become expanded, and completes its movement by contracting to once again occupy a single node. For an expanded particle, we denote the node it last expanded into as its {\it head} and the other node it occupies as its {\it tail}; for a contracted particle, the single node it occupies is both its head and its tail.

Two particles occupying adjacent nodes are said to be {\it neighbors} and are connected by a {\it bond}. These bonds both ensure that the overall particle system remains connected as well as providing a mechanism for exchanging information between particles. In order to maintain connectivity as they move, neighboring particles coordinate their motion in a \emph{handover}, which can occur in two ways. A contracted particle $p$ can initiate a handover by expanding into a node occupied by an expanded neighbor $q$, ``pushing'' $q$ and forcing it to contract. Alternatively, an expanded particle $q$ can initiate a handover by contracting, ``pulling'' a contracted neighbor $p$ to the node it is vacating, thereby forcing $p$ to expand. Figures~\ref{fig:model}b and~\ref{fig:model}c illustrate two particles performing a handover.

\begin{figure}
\centering
\includegraphics[width=.9\textwidth]{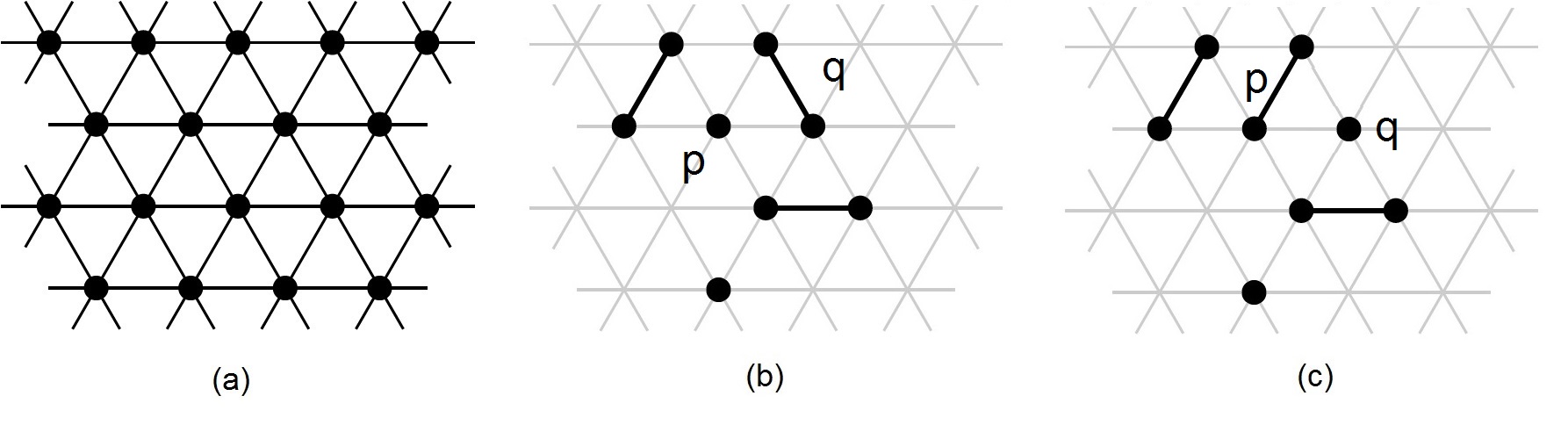}
\vspace{-.15in}
\caption{\small (a) shows a section of $\Geqt$, where nodes of $\Geqt$ are shown as black circles. (b) shows five particles on $\Geqt$; the underlying graph $\Geqt$ is depicted as a gray mesh; a contracted particle is depicted as a single black circle and an expanded particle is depicted as two black circles connected by an edge. (c) depicts the resulting configuration after a handover was performed by particles $p$ and $q$ in (b).}
\vspace{-.2in}
\label{fig:model}
\end{figure}

Particles are {\it anonymous}, but each keeps a collection of uniquely labeled {\it ports} corresponding to the edges incident to the node(s) it occupies. Bonds between neighboring particles are formed through ports that face each other.
The particles are assumed to have a common {\it chirality}, meaning they share the same notion of {\it clockwise (CW) direction}. This allows each particle to label its ports counting in the clockwise direction; without loss of generality, we assume each particle labels its head and tail ports from $0$ to $5$.
However, particles may have different offsets for their port labels, and thus do not share a common sense of orientation.
Each particle has a constant-size, local memory for which both it and its neighbors have read and write access. Particles can communicate by writing into each other's memories. Due to the limitation of constant-size memory, particles have no knowledge of the total number of particles in the system, nor do they have any approximation of this value.
We assume that any conflicts of movement or simultaneous memory writes are resolved arbitrarily, so that at most one particle writes to any memory location or moves into an empty position at any given time.

The \emph{configuration} $C$ of the particle system at the beginning of time $t$ consists of (1) the nodes in $\Geqt$ occupied by the object and the set of particles, and (2) the current state of each particle, including whether it is expanded or contracted, its port labeling, and the contents of its local memory.

Following the standard asynchronous model of computation~\cite{Lynch1996}, we assume that the system progresses through a sequence of atomic \emph{activations} of individual particles. When activated, a particle can perform a bounded amount of computation involving its local memory and the memories of its neighbors and at most one movement. A classical result under this model is that for any asynchronous concurrent execution of atomic activations, there exists a sequential ordering of the activations which produces the same end configuration, provided conflicts arising from the concurrent execution are resolved (as they are in our scenario). We assume the resulting activation sequence is \emph{fair}\footnote{We will see this notion of fairness is sufficient to prove the desired runtime for our algorithm; no further assumptions regarding the distribution of the activation sequence are necessary.}, i.e., for each particle $p$ and any time $t$, $p$  will eventually be activated at some time $t' > t$. An \emph{asynchronous round} is complete once every particle has been activated at least once.

\vspace{-.1in}
\subsection{Universal Coating Problem} \label{sec:problem}
\vspace{-.05in}
In the {\em universal coating problem} we consider an instance $(P,O)$ where $P$ represents the particle system and $O$ represents the fixed object to be coated. Let $n = |P|$ be the number of particles in the system, $V(P)$ be the set of nodes occupied by $P$, and $V(O)$ be the set of nodes occupied by $O$ (when clear from the context, we may omit the $V(\cdot)$ notation).
For any two nodes $v,w \in \Veqt$, the \emph{distance} $d(v,w)$ between $v$ and $w$ is the length of the shortest path in $\Geqt$ from $v$ to $w$. The distance $d(v,U)$ between a $v \in \Veqt$ and $U \subseteq \Veqt$ is defined as $\min_{w \in U} d(v,w)$. Define \emph{layer $i$} to be the set of nodes that have a distance $i$ to the object, and let $B_i$ be the number of nodes in layer $i$.
An instance is {\em valid} if the following properties hold:
\begin{enumerate}
\item The particles are all contracted and are initially in the \emph{idle} state.

\item The subgraphs of $\Geqt$ induced by $V(O)$ and $V(P) \cup V(O)$, respectively, are connected, i.e., there is a single object and the particle system is connected to the object.

\item The subgraph of $\Geqt$ induced by $\Veqt \setminus V(O)$ is connected, i.e., the object $O$ has no holes.\footnote{If $O$ does contain holes, we consider the subset of particles in each connected region of $\Veqt \setminus V(O)$ separately.}

\item $\Veqt \setminus V(O)$ is $2(\lceil\frac{n}{B_1}\rceil +1)$-connected, i.e., $O$ cannot form \emph{tunnels} of width less than $2(\lceil\frac{n}{B_1}\rceil +1)$.
\end{enumerate}

Note that a width of at least $2 \lceil\frac{n}{B_1}\rceil$ is needed to guarantee that the object can be evenly coated. The coating of narrow tunnels requires specific technical mechanisms that complicate the protocol without contributing to the basic idea of coating, so we ignore such cases in favor of simplicity.

A configuration $C$ is \emph{legal} if and only if all particles are contracted and
\[\min_{v \in \Veqt \setminus (V(P) \cup V(O))} d(v,V(O)) \ge \max_{v \in V(P)} d(v,V(O)),\]
meaning that all particles are as close to the object as possible or {\em coat $O$ as evenly as possible}.
A configuration $C$ is said to be \emph{stable} if no particle in $C$ ever performs a state change or movement.
An algorithm \emph{solves} the universal coating problem if, starting from any valid instance, it reaches a {\em stable legal configuration} in a finite number of rounds.

\vspace{-.1in}
\subsection{Related work} \label{sec:relwork}
\vspace{-.05in}
Many approaches have been proposed with potential applications in smart coating; these can be categorized as active and passive systems. In passive systems, particles move based only on their structural properties and interactions with their environment, or have only limited computational ability but lack control of their motion. Examples include population protocols~\cite{Angluin2006} as well as molecular computing models such as DNA self-assembly systems (see, e.g., the surveys in~\cite{Doty2012,Patitz2014,Woods2015}) and slime molds~\cite{Bonifaci2012,Li2010}.

Our focus, however, is on active systems, in which computational particles control their actions and motions to complete specific tasks. Coating has been extensively studied in the area of \textit{swarm robotics}, but not commonly treated as a stand-alone problem; it is instead examined as part of \emph{collective transport} (e.g.,~\cite{Wilson2014}) or \emph{collective perception} (e.g., see respective section of~\cite{Brambilla2013}).
Some research focuses on coating objects as an independent task under the name of \emph{target surrounding} or \emph{boundary coverage}. The techniques used in this context include stochastic robot behaviors~\cite{Kumar2014,Pavlic2016}, rule-based control mechanisms~\cite{Blazovics2012-CFC} and potential field-based approaches~\cite{Blazovics2012-LF}.
While the analytic techniques developed in swarm robotics are somewhat relevant to this work, many such systems assume more computational power and movement capabilities than the model studied in this work does.
Michail and Spirakis recently proposed a model~\cite{Michail2016} for network construction inspired by population protocols~\cite{Angluin2006}. The population protocol model is related to self-organizing particle systems, but is different in that agents (corresponding to our particles) can move freely in space and establish connections at any time. It would, however, be possible to adapt their approach to study coating problems under the population protocol model.

In the context of molecular programming, our model most closely relates to the {\it nubot} model by Woods et al.~\cite{Woods2013,Chen2015}, which seeks to provide a framework for rigorous algorithmic research on self-assembly systems composed of active molecular components, emphasizing the interactions between molecular structure and active dynamics. This model shares many characteristics of our amoebot model (e.g., space is modeled as a triangular grid, nubot monomers have limited computational abilities, and there is no global orientation) but differs in that nubot monomers can replicate or die and can perform coordinated rigid body movements.
These additional capabilities prohibit the direct translation of results under the nubot model to our amoebot model; the latter provides a framework for future, large-scale swarm robotic systems of computationally limited particles (each possibly at the nano- or micro-scale) with only local control and coordination mechanisms, where these capabilities would likely not apply.

Finally, in~\cite{Derakhshandeh2017} we presented our Universal Coating algorithm and proved its correctness. We also showed it to be worst-case work-optimal, where work is measured in terms of number of particle movements.

\vspace{-.1in}
\subsection{Our Contributions} \label{sec:contr}
\vspace{-.05in}
In this paper we continue the analysis of the {\it Universal Coating algorithm} introduced in~\cite{Derakhshandeh2017}.
As our main contribution in this paper, we investigate the runtime of our algorithm and prove that our algorithm terminates within a {\it linear number of rounds} with high probability. This result relies, in part, on an update to the leader election protocol used in~\cite{Derakhshandeh2017} which is fully defined and analyzed in~\cite{Daymude2017}.
We also present a matching linear lower bound for any {\it local-control} coating algorithm (i.e., one which uses only local information in its execution) that holds with high probability. We use this lower bound to show a {\it linear lower bound on the competitive gap} between fully local coating algorithms and coating algorithms that rely on global information, which implies that our algorithm is also optimal in a competitive sense.
We then present some simulation results demonstrating that in practice the competitive ratio of our algorithm is often much better than linear.

\vspace{-.1in}
\subsubsection{Overview}
In Section~\ref{sec:algo}, we again present the algorithm introduced in~\cite{Derakhshandeh2017}. We then present a comprehensive formal runtime analysis of our algorithm, by first presenting some lower bounds on the competitive ratio of any local-control algorithm in Section~\ref{sec:performance}, and then proving that our algorithm has a runtime of $\mathcal{O}(n)$ rounds w.h.p.~in Section~\ref{sec:WCruntime}, which matches our lower bounds.

\vspace{-.1in}
\section{Universal Coating Algorithm} \label{sec:algo}
\vspace{-.05in}
In this section, we summarize the Universal Coating algorithm introduced in~\cite{Derakhshandeh2017} (see~\cite{Derakhshandeh2017} for a detailed description). This algorithm is constructed by combining a number of asynchronous primitives, which are integrated seamlessly without any underlying synchronization. The {\it spanning forest} primitive organizes the particles into a spanning forest, which determines the movement of particles while preserving system connectivity; the {\it complaint-based coating} primitive coats the first layer by bringing any particles not yet touching the object into the first layer while there is still room; the {\it general layering} primitive allows each layer $i$ to form only after layer $i-1$ has been completed, for $i\geq 2$; and the {\it node-based leader election} primitive elects a node in layer 1 whose occupant becomes the leader particle, which is used to trigger the general layering process for higher layers.

\vspace{-.1in}
\subsection{Preliminaries} \label{sec:pre}
\vspace{-.05in}
We define the set of \emph{states} that a particle can be in as \emph{idle}, \emph{follower}, \emph{root}, and \emph{retired}. In addition to its state, a particle maintains a constant number of other flags, which in our context are constant size pieces of information visible to neighboring particles. A flag $x$ owned by some particle $p$ is denoted by $p.x$.
Recall that a \emph{layer} is the set of nodes $V \subseteq \Veqt$ that are equidistant to the object $O$. A particle keeps track of its current layer number in $p.layer$. In order to respect the constant-size memory constraint of particles, we take all layer numbers modulo $4$.
Each root particle $p$ has a flag $p.down$ which stores a port label pointing to a node of the object if $p.layer=1$, and to an occupied node adjacent to its head in layer $p.layer - 1$ if $p.layer>1$.
We now describe the coating primitives in more detail.

\vspace{-.1in}
\subsection{Coating Primitives} \label{subsec:CoatingPrimitives}
\vspace{-.05in}
The \textbf{\emph{spanning forest primitive}} (Algorithm~\ref{alg:spanningForestAlgorithm}) organizes the particles into a spanning forest $\mathcal{F}$, which yields a straightforward mechanism for particles to move while preserving connectivity (see~\cite{Daymude2016,Derakhshandeh2015-nanocom} for details).
Initially, all particles are \emph{idle}. A particle $p$ touching the object changes its state to \emph{root}. For any other idle particle $p$, if $p$ has a root or a follower in its neighborhood, it stores the direction to one of them in $p.parent$, changes its state to \emph{follower}, and generates a complaint flag; otherwise, it remains idle.
A follower particle $p$ uses handovers to follow its parent and updates the direction $p.parent$ as it moves in order to maintain the same parent in the tree (note that the particular particle at $p.parent$ may change due to $p$'s parent performing a handover with another of its children). In this way, the trees formed by the parent relations stay connected, occupy only the nodes they covered before, and do not mix with other trees.
A root particle $p$ uses the flag $p.dir$ to determine its movement direction. As $p$ moves, it updates $p.dir$ so that it always points to the next position of a clockwise movement around the object.
For any particle $p$, we call the particle occupying the position that $p.parent$ resp. $p.dir$ points to the \emph{predecessor} of $p$. If a root particle does not have a predecessor, we call it a \emph{super-root}.

\begin{algorithm}
\caption{Spanning Forest Primitive}
\label{alg:spanningForestAlgorithm}
    A particle $p$ acts depending on its state as described below: \\
    \begin{tabularx}{\textwidth}{lX}
        \textbf{idle}: &
        If $p$ is adjacent to the object $O$, it becomes a \emph{root} particle, makes the current node it occupies a {\it leader candidate position}, and starts running the leader election algorithm.
        If $p$ is adjacent to a \emph{retired} particle, $p$ also becomes a \emph{root} particle.
        If a neighbor $p'$ is a root or a follower, $p$ sets the flag $p.parent$ to the label of the port to $p'$, puts a \emph{complaint flag} in its local memory, and becomes a \emph{follower}.
        If none of the above applies, $p$ remains idle.
        \\

        \textbf{follower}: &
        If $p$ is contracted and adjacent to a retired particle or to $O$, then $p$ becomes a \emph{root} particle.
    If $p$ is contracted and has an expanded parent, then $p$ initiates {\sc Handover$(p)$} (Algorithm~\ref{alg:handover}); otherwise, if $p$ is expanded, it considers the following two cases: $(i)$ if $p$ has a contracted child particle $q$, then $p$ initiates {\sc Handover$(p)$}; $(ii)$ if $p$ has no children and no idle neighbor, then $p$ contracts.
        Finally, if $p$ is contracted, it runs the function {\sc ForwardComplaint$(p,p.parent)$} (Algorithm~\ref{alg:complaint}).
        \\

        \textbf{root}: &
        If particle $p$ is in layer 1, $p$ participates in the leader election process.
        If $p$ is contracted, it first executes {\sc MarkerRetiredConditions$(p)$} (Algorithm~\ref{alg:retiredCondition}) and becomes \emph{retired}, and possibly also a \emph{marker}, accordingly. If $p$ does not become retired, then if it has an expanded root in $p.dir$, it initiates {\sc Handover$(p)$}; otherwise, $p$ calls {\sc LayerExtension$(p)$} (Algorithm~\ref{alg:boundaryDirectionAlgorithm}).
        If $p$ is expanded, it considers the following two cases: $(i)$ if $p$ has a contracted child, then $p$ initiates {\sc Handover$(p)$}; $(ii)$ if $p$ has no children and no idle neighbor, then $p$ contracts.
        Finally, if $p$ is contracted, it runs {\sc ForwardComplaint$(p,p.dir)$}.
        \\

        \textbf{retired}: &
        $p$ clears a potential complaint flag from its memory and performs no further action.
    \end{tabularx}
\end{algorithm}

The \textbf{\emph{complaint-based coating primitive}} is used for the coating of the first layer. Each time a particle $p$ holding at least one complaint flag is activated, it forwards one to its predecessor as long as that predecessor holds less than two complaint flags. We allow each particle to hold up to two complaint flags to ensure that a constant size memory is sufficient for storing the complaint flags and so the flags quickly move forward to the super-roots. A contracted super-root $p$ only expands to $p.dir$ if it holds at least one complaint flag, and when it expands, it consumes one of these complaint flags. All other roots $p$ move towards $p.dir$ whenever possible (i.e., no complaint flags are required) by performing a handover with their predecessor (which must be another root) or a successor (which is a root or follower of its tree), with preference given to a follower so that additional particles enter layer 1. As we will see, these rules ensure that whenever there are particles in the system that are not yet at layer 1, eventually one of these particles will move to layer 1, unless layer 1 is already completely filled with contracted particles.

\begin{algorithm}
\caption{{\sc Handover}$(p)$}
\label{alg:handover}
    \begin{algorithmic}[1]
      \If{$p$ is expanded}
            \If {$p.layer = 1$ and $p$ has a follower child $q$ such that $q.parent$ points to the tail \item[\hspace{3em} of $p$]}
                \If {$q$ is contracted}
                \State $p$ initiates a handover with particle $q$
                \EndIf
            \ElsIf {$p$ has a contracted (follower or root) child $q$ such that $q.parent$ points \item[\hspace{3em} to the tail of $p$]}
        \State $p$ initiates a handover with particle $q$
            \EndIf
        \ElsIf {$p$ has an expanded parent $q$ or the position in $p.dir$ is occupied by an expanded root $q$}
    \State $p$ initiates a handover with particle $q$
      \EndIf
    \end{algorithmic}
\end{algorithm}

\begin{algorithm}
\caption{{\sc ForwardComplaint$(p,i)$}}
\label{alg:complaint}
    \begin{algorithmic}[1]
        \If {$p$ holds at least one complaint flag \textbf{and} the particle $q$ adjacent to $p$ in direction $i$ holds less than two complaint flags}
        \State $p$ forwards one complaint flag to $q$
        \EndIf
    \end{algorithmic}
\end{algorithm}

The \textbf{\emph{leader election primitive}} runs during the complaint-based coating primitive to elect a node in layer 1 as the leader position. This primitive is similar to the algorithm presented in~\cite{Daymude2016} with the difference that leader candidates are nodes instead of static particles (which is important because in our case particles may still move during the leader election primitive). The primitive only terminates once all positions in layer 1 are occupied. Once the leader position is determined, all positions in layer 1 are filled by contracted particles and whatever particle currently occupies that position becomes the \emph{leader}.
This leader becomes a marker particle, marking a neighboring position in the next layer as a {\it marked position} which determines a starting point for layer 2, and becomes \emph{retired}.
Once a contracted root $p$ has a retired particle in the direction $p.dir$, it retires as well, which causes the particles in layer 1 to become retired in counter-clockwise order. At this point, the general layering primitive becomes active, which builds subsequent layers until there are no longer followers in the system.
If the leader election primitive does not terminate (which only happens if $n<B_1$ and layer 1 is never completely filled), then the complaint flags ensure that the super-roots eventually stop, which eventually results in a stable legal coating.

\begin{algorithm}
\caption{{\sc LayerExtension$(p)$}}
\label{alg:boundaryDirectionAlgorithm}
    \begin{algorithmic}[1]
        \Statex{\bf Calculating $p.layer$, $p.down$ and $p.dir$}
        \State The layer number of any node occupied by the object is equal to 0.
        \State Let $q$ be any neighbor of $p$ with smallest layer number (modulo $4$).
        \State $p.down \gets \mbox{$p$'s label for port leading to } q$
        \State $p.layer=(q.layer+1)\mod \: 4 $
        \State {\sc Clockwise$(p, p.down)$} \Comment{Computes CW \& CCW directions}
        \If {$p.layer$ is {\em odd}}
            \State $p.dir \gets p.CW$
        \Else
            \State $p.dir \gets p.CCW$
        \EndIf

        \Statex
        \Statex \textbf{Extending layer $p.layer$}
        \If {the position at $p.dir$ is unoccupied, and either $p$ is not on the first layer or $p$ holds a complaint flag}
            \State $p$ expands in direction $p.dir$
            \State $p$ consumes a complaint flag, if it holds one
        \EndIf
    \end{algorithmic}
\end{algorithm}

In the \textbf{\emph{general layering primitive}}, whenever a follower is adjacent to a retired particle, it becomes a root. Root particles continue to move along positions of their layer in a clockwise (if the layer number is odd) or counter-clockwise (if the layer number is even) direction until they reach either the marked position of that layer, a retired particle in that layer, or an empty position of the previous layer (which causes them to change direction). Complaint flags are no longer needed to expand into empty positions. Followers follow their parents as before. A contracted root particle $p$ may retire if: (i) it occupies the marked position and the marker particle in the lower layer tells it that all particles in that layer are retired (which it can determine locally), or (ii) it has a retired particle in the direction $p.dir$. Once a particle at a marked position retires, it becomes a marker particle and marks a neighboring position in the next layer as a marked position.

\begin{algorithm}
\caption{{\sc Clockwise$(p,i)$}}
\label{alg:clockwise}
    \begin{algorithmic}[1]
        \State $j \gets i$, $k \gets i$
        \While{edge $j$ is incident to the object or to a retired particle with layer number $p.layer-1$}
            \State $j \gets (j - 1) \mod \; 6$
        \EndWhile
        \State $p.CW \gets j$
        \While{edge $k$ is incident to the object or to a retired particle with layer number $p.layer-1$}
            \State $k \gets (k + 1) \mod \; 6$
        \EndWhile
        \State $p.CCW \; \gets \; k$
    \end{algorithmic}
\end{algorithm}

\begin{algorithm}
\caption{{\sc MarkerRetiredConditions$(p)$}}
\label{alg:retiredCondition}
    \begin{algorithmic}[1]
        \Statex{\bf First Marker Condition:}
        \If {$p$ is the \emph{leader}}
            \State $p$ becomes a {\em retired} particle
            \State \parbox[t]{\dimexpr\linewidth-\algorithmicindent}{$p$ sets the flag $p.marker$ to be the label of a port leading to a node guaranteed not to be in layer 1 --- e.g., by taking the average direction of $p$'s two neighbors in layer 1 (by now complete)\strut}
        \EndIf

        \Statex
        \Statex{\bf Extending Layer Markers:}
        \If {$p$ is connected to a marker $q$ and the port $q.marker$ points towards $p$}
            \If {both $q.CW$ and $q.CCW$ are retired}
                \State $p$ becomes a \emph{retired} particle
                \State \parbox[t]{\dimexpr\linewidth-3em}{$p$ sets the flag $p.marker$ to the label of the port opposite the port connecting $p$ to $q$\strut}
            \EndIf
        \EndIf

        \Statex
        \Statex{\bf Retired Condition:}
        \If{the node in direction $p.dir$ is occupied by a retired particle}
            \State $p$ becomes a \emph{retired} particle
        \EndIf
    \end{algorithmic}
\end{algorithm}

\vspace{-.1in}
\section{Lower Bounds} \label{sec:performance}
\vspace{-.05in}
Recall that a \emph{round} is over once every particle in $P$ has been activated at least once.
The \emph{runtime} $T_{\mathcal{A}}(P,O)$ of a coating algorithm $\mathcal{A}$ is defined as the worst-case number of rounds (over all sequences of particle activations) required for $\mathcal{A}$ to solve the coating problem $(P,O)$. Certainly, there are instances $(P,O)$ where every coating algorithm has a runtime of $\Omega(n)$ (see Lemma~\ref{lem:lowerbound}), though there are also many other instances where the coating problem can be solved much faster. Since a worst-case runtime of $\Omega(n)$ is fairly large and therefore not very helpful to distinguish between different coating algorithms, we intend to study the runtime of coating algorithms relative to the best possible runtime.

\begin{lemma}
\label{lem:lowerbound}
The worst-case runtime required by any local-control algorithm to solve the universal coating problem is $\Omega(n)$.
\end{lemma}
\begin{proof}
\begin{figure}
    \centering
    \includegraphics[width=60mm]{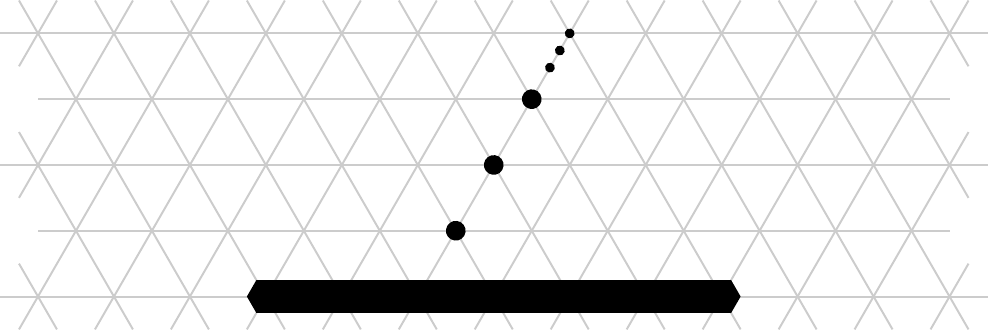}
    \caption[Lower Bound]{\small Worst-case configuration concerning number of rounds. There are $n$ particles (black dots) in a line connected to the surface via a single particle $p_1$.}
    \vspace{-.15in}
    \label{fig:lowerbound}
\end{figure}

Assume the particles $p_1, \ldots, p_n$ form a single line of $n$ particles connected to the surface via $p_1$ (Figure~\ref{fig:lowerbound}). Suppose $B_1>n$. Since  $d(p_n,O)=n$, it will take $\Omega(n)$ rounds in the worst-case (requiring $\Theta(n)$ movements) until $p_n$ touches the object's surface. This worst-case can happen, for example, if $p_n$ performs no more than one movement (either an expansion or a contraction) per round.
\qed
\end{proof}

Unfortunately, a large lower bound also holds for the competitiveness of any local-control algorithm. A coating algorithm $\mathcal{A}$ is called {\it $c$-competitive} if for any valid instance $(P,O)$,
\[\E[T_{\mathcal{A}}(P,O)] \leq c \cdot \OPT(P,O) + K\]
where $\OPT(P,O)$ is the minimum runtime needed to solve the coating problem $(P,O)$ and $K$ is a value independent of $(P,O)$.

\begin{theorem}
\label{thm:competitiveness}
Any local-control algorithm that solves the universal coating problem has a competitive ratio of $\Omega(n)$.
\end{theorem}
\begin{proof}
We construct an instance of the coating problem $(P,O)$ which can be solved by an optimal algorithm in $\mathcal{O}(1)$ rounds, but requires any local-control algorithm $\Omega(n)$ times longer. Let $O$ be a straight line of arbitrary (finite) length, and let $P$ be a set of particles which entirely occupy layer 1, with the exception of one missing particle below $O$ equidistant from its sides and one additional particle above $O$ in layer 2 equidistant from its sides (see Figure~\ref{fig:compgapconfig}).

An optimal algorithm could move the particles to solve the coating problem for the given example in $\mathcal{O}(1)$ rounds, as in Figure~\ref{fig:compgapopt}. Note that the optimal algorithm always maintains the connectivity of the particle system, so its runtime is valid even under the constraint that any connected component of particles must stay connected. However, for our local-control algorithms we allow particles to disconnect from the rest of the system.

\begin{figure}
    \centering
    \vspace{-.1in}
    \includegraphics[width=.6\textwidth]{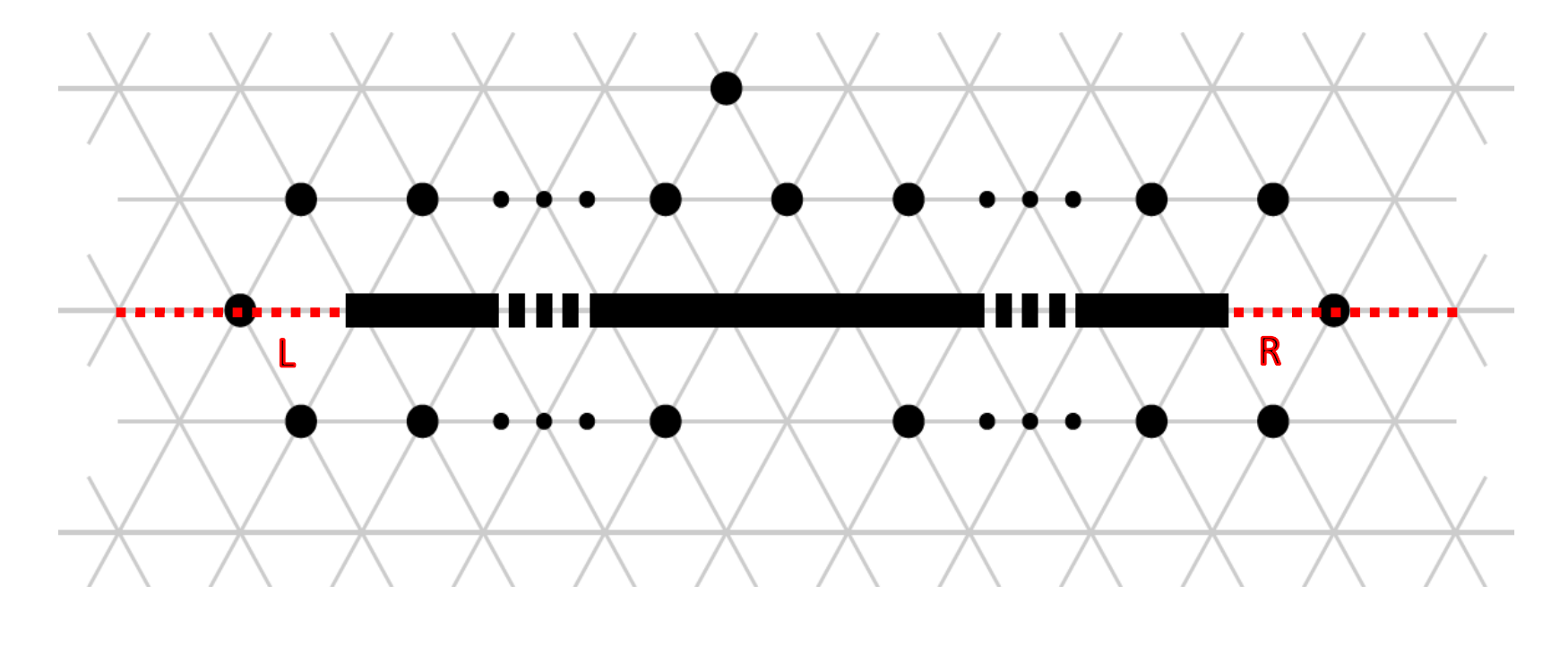}
    \vspace{-.15in}
    \caption[Competitiveness]{\small The object occupies a straight line in $\Geqt$. The particles are all contracted and occupy the positions around the object, with the exception that there is one unoccupied node below the object and one extra particle above the object. Borders $L$ and $R$ are shown as red lines.}
    \vspace{-.15in}
    \label{fig:compgapconfig}
\end{figure}

Now consider an arbitrary local-control algorithm $A$ for the coating problem. Given a round $r$, we define the \emph{imbalance} $\phi_L(r)$ at border $L$ as the net number of particles that have crossed $L$ from the top of $O$ to the bottom until round $r$; similarly, the imbalance $\phi_R(r)$ at border $R$ is defined to be the net number of particles that have crossed $R$ from the bottom of $O$ to the top until round $r$.

\begin{figure}
    \centering
    \vspace{-.1in}
    \includegraphics[width=\textwidth]{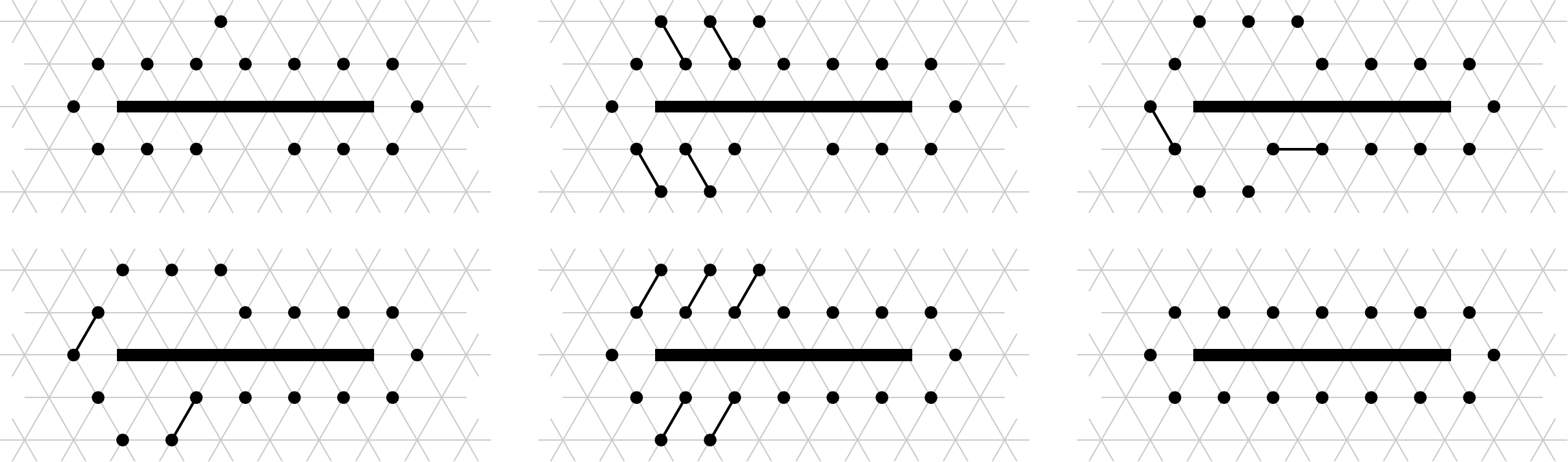}
    \vspace{-.15in}
    \caption[Optimal Algorithm]{\small Each subfigure represents the configuration of the system at the beginning of a round, and are ordered from left to right, top to bottom. After 5 rounds (i.e., at the beginning of the sixth round) the object is coated. Note that the implied algorithm can be adapted to any length of the object and always requires only 5 rounds to solve the coating problem.}
    \vspace{-.15in}
    \label{fig:compgapopt}
\end{figure}

Certainly, there is an activation sequence in which information and particles can only travel a distance of up to $n/4$ nodes towards $L$ or $R$ within the first $n/4$ rounds. Hence, for any $r \leq n/4$, the probability distributions of $\phi_L(r)$ and $\phi_R(r)$ are independent of each other. Additionally, particles up to a distance of $n/4$ from $L$ and $R$ cannot distinguish between which border they are closer to, since the position of the gap is equidistant from the borders. This symmetry also implies that $\Pr[\phi_L(r)=k]=\Pr[\phi_R(r)=k]$ for any integer $k$. Let us focus on round $r=n/4$. We distinguish between the following cases.

\begin{case}
$\phi_L(n/4)=\phi_R(n/4)$. Then there are more particles than positions in layer 1 above $O$, so the coating problem cannot be solved yet.
\end{case}

\begin{case}
$\phi_L(n/4)\neq\phi_R(n/4)$. From our insights above we know that for any two values $k_1$ and $k_2$, $\Pr[\phi_L(n/4)=k_1$ and $\phi_R(n/4)=k_2] = \Pr[\phi_L(n/4)=k_2$ and $\phi_R(n/4)=k_1]$. Hence, the cumulative probability of all outcomes where $\phi_L(n/4)<\phi_R(n/4)$ is equal to the cumulative probability of all outcomes where $\phi_L(n/4)>\phi_R(n/4)$. If $\phi_L(n/4)<\phi_R(n/4)$, then there are again more particles than positions in layer 1 above $O$, so the coating problem cannot be solved yet.
\end{case}

Thus, the probability that $\mathcal{A}$ has not solved the coating problem after $n/4$ rounds is at least $1/2$, and therefore $\E[T_{\mathcal{A}}(P,O)] \geq 1/2 \cdot n/4 = n/8$. Since, on the other hand, $\OPT = \mathcal{O}(1)$, we have established a linear competitive ratio.
\qed
\end{proof}

Therefore, even the competitive ratio can be very high in the worst case. We will revisit the notion of competitiveness in Section~\ref{sec:experimental}.

\vspace{-.1in}
\section{Worst-Case Number of Rounds} \label{sec:WCruntime}
\vspace{-.05in}
In this section, we show that our algorithm solves the coating problem within a linear number of rounds w.h.p.
\footnote{This version of the paper reflects what was submitted to the DNA22 Special Issue of the journal Natural Computing, and updates the logical structure of this section from its original publication in DNA22.}.
We start with some basic notation in Section~\ref{sec:preliminaryanalysis}. Section~\ref{sec:parallel} presents a simpler synchronous parallel model for particle activations that we can use to analyze the worst-case number of rounds. Section~\ref{sec:firstlayer} presents the analysis of the number of rounds required to coat the first layer. Finally, in Section~\ref{sec:higherlayers}, we analyze the number of rounds required to complete all other coating layers, once layer 1 has been completed.

\vspace{-.1in}
\subsection{Preliminaries} \label{sec:preliminaryanalysis}
\vspace{-.05in}
We start with some notation. Recall that $B_i$ denotes the number of nodes of $\Geqt$ at distance $i$ from object $O$ (i.e., the number of nodes in layer $i$). Let $N$ be the the layer number of the final layer for $n$ particles (i.e., $N$ satisfies $\sum_{j=1}^{N-1}B_j<n\leq \sum_{j=1}^{N}B_j$).
Layer $i$ is said to be \emph{complete} if every node in layer $i$ is occupied by a contracted retired particle (for $i < N$), or if all particles have reached their final position, are contracted, and never move again (for $i = N$).

Given a configuration $C$, we define a directed graph $A(C)$ over all nodes in $\Geqt$ occupied by \emph{active} (follower or root) particles in $C$. For every expanded active particle $p$ in $C$, $A(C)$ contains a directed edge from the tail to the head of $p$. For every follower $p$, $A(C)$ has a directed edge from the head of $p$ to $p.parent$. For the purposes of constructing $A(C)$, we also define parents for root particles: a root particle $p$ sets $p.parent$ to be the active particle $q$ occupying the node in direction $p.dir$ once $p$ has performed its first handover expansion with $q$. For every root particle $p$, $A(C)$ has a directed edge from the head of $p$ to $p.parent$, if it exists.
Certainly, since every node has at most one outgoing edge in $A(C)$, the nodes of $A(C)$ form either a collection of disjoint trees or a ring of trees. A ring of trees may occur in any layer, but only temporarily; the leader election primitive ensures that a leader emerges and retires in layer 1 and marker particles emerge and retire in higher layers, causing the ring in $A(C)$ to break.
The super-roots defined in Section~\ref{subsec:CoatingPrimitives} correspond to the roots of the trees in $A(C)$.

A \emph{movement} executed by a particle $p$ can be either a \emph{sole contraction} in which $p$ contracts and leaves a node unoccupied, a \emph{sole expansion} in which $p$ expands into an adjacent unoccupied node, a \emph{handover contraction with $p'$} in which $p$ contracts and forces its contracted neighbor $p'$ to expand into the node it vacates, or a \emph{handover expansion with $p'$} in which $p$ expands into a node currently occupied by its expanded neighbor $p'$, forcing $p'$ to contract.

\vspace{-.1in}
\subsection{From asynchronous to parallel schedules} \label{sec:parallel}
\vspace{-.05in}
In this section, we show that instead of analyzing our algorithm for asynchronous activations of particles, it suffices to consider a much simpler model of parallel activations of particles.
Define a \emph{movement schedule} to be a sequence of particle system configurations $(C_0, \ldots, C_t)$.

\begin{definition}
\label{defn:parallelschedule}
A movement schedule $(C_0, \ldots, C_t)$ is called a \emph{parallel schedule} if each $C_i$ is a valid configuration of a connected particle system (i.e., each particle is either expanded or contracted, and every node of $\Geqt$ is occupied by at most one particle) and for every $i \geq 0,~C_{i+1}$ is reached from $C_i$ such that for every particle $p$ one of the following properties holds:
\begin{enumerate}
    \item $p$ occupies the same node(s) in $C_i$ and $C_{i+1}$,
    \item $p$ expands into an adjacent node that was empty in $C_i$,
    \item $p$ contracts, leaving the node occupied by its tail empty in $C_{i+1}$, or
    \item $p$ is part of a handover with a neighboring particle $p'$.
\end{enumerate}
\end{definition}

While these properties allow at most one contraction or expansion per particle in moving from $C_i$ to $C_{i+1}$, multiple particles may move in this time.

Consider an arbitrary fair asynchronous activation sequence $A$ for a particle system, and let $C_i^{(A)}$, for $0 \leq i \leq t$, be the particle system configuration at the end of asynchronous round $i$ in $A$ if each particle moves according to Algorithm~\ref{alg:spanningForestAlgorithm}.
A \emph{forest schedule} $\mathcal{S} = (A, (C_0, \ldots, C_t))$ is a parallel schedule $(C_0, \ldots, C_t)$ with the property that $A(C_0)$ is a forest of one or more trees, and each particle $p$ follows the unique path $P_p$ which it would have followed according to $A$, starting from its position in $C_0$. This implies that $A(C_i)$ remains a forest of trees for every $1 \leq i \leq t$.
A forest schedule is said to be \emph{greedy} if all particles perform movements according to Definition~\ref{defn:parallelschedule} in the direction of their unique paths whenever possible.

We begin our analysis with a result that is critical to both describing configurations of particles in greedy forest schedules and quantifying the amount of progress greedy forest schedules make over time. Specifically, we show that if a forest's configuration is ``well-behaved'' at the start, then it remains so throughout its greedy forest schedule, guaranteeing that progress is made once every two configurations.

\begin{lemma}
\label{lem:expparentchild}
Given any fair asynchronous activation sequence $A$, consider any greedy forest schedule $(A, (C_0, \ldots, C_t))$. If every expanded parent in $C_0$ has at least one contracted child, then every expanded parent in $C_i$ also has at least one contracted child, for $1 \leq i \leq t$.
\end{lemma}
\begin{proof}
Suppose to the contrary that $C_i$ is the first configuration that contains an expanded parent $p$ which has all expanded children. We consider all possible expanded and contracted states of $p$ and its children in $C_{i-1}$ and show that none of them can result in $p$ and its children all being expanded in $C_i$.
First suppose $p$ is expanded in $C_{i-1}$; then by supposition, $p$ has a contracted child $q$. By Definition~\ref{defn:parallelschedule}, $q$ cannot perform any movements with its children (if they exist), so $p$ performs a handover contraction with $q$, yielding $p$ contracted in $C_i$, a contradiction.
So suppose $p$ is contracted in $C_{i-1}$. We know $p$ will perform either a handover with its parent or a sole expansion in direction $p.dir$ since it is expanded in $C_i$ by supposition. Thus, any child of $p$ in $C_{i-1}$ --- say $q$ --- does not execute a movement with $p$ in moving from $C_{i-1}$ to $C_i$. Instead, if $q$ is contracted in $C_{i-1}$ then it remains contracted in $C_i$ since it is only permitted to perform a handover with its unique parent $p$; otherwise, if $q$ is expanded, it performs either a sole contraction if it has no children or a handover with one of its contracted children, which it must have by supposition. In either case, $p$ has a contracted child in $C_i$, a contradiction.

As a final observation, two trees of the forest may ``merge'' when the super-root $s$ of one tree performs a sole expansion into an unoccupied node adjacent to a particle $q$ of another tree. However, since $s$ is a root and thus only defines $q$ as its parent after performing a handover expansion with it, the lemma holds in this case as well.
\qed
\end{proof}

For any particle $p$ in a configuration $C$ of a forest schedule, we define its \emph{head distance} $d_h(p,C)$ (resp., \emph{tail distance} $d_t(p,C)$) to be the number of edges along $P_p$ from the head (resp., tail) of $p$ to the end of $P_p$. Depending on whether $p$ is contracted or expanded, we have $d_h(p,C) \in \{d_t(p,C), d_t(p,C)-1\}$.
For any two configurations $C$ and $C'$ and any particle $p$, we say that $C$ \emph{dominates $C'$ w.r.t.~$p$}, denoted $C(p) \succeq C'(p)$, if and only if $d_h(p,C) \leq d_h(p,C')$ and $d_t(p,C) \leq d_t(p,C')$. We say that $C$ \emph{dominates} $C'$, denoted $C \succeq C'$, if and only if $C$ dominates $C'$ with respect to every particle. Then it holds:

\begin{lemma}
\label{lem:forestdom}
Given any fair asynchronous activation sequence $A$ which begins at an initial configuration $C_0^{(A)}$ in which every expanded parent has at least one contracted child, there is a greedy forest schedule $\mathcal{S} = (A, (C_0, \ldots, C_t))$ with $C_0 = C_0^{(A)}$ such that $C_i^{(A)} \succeq C_i$ for all $0 \leq i \leq t$.
\end{lemma}
\begin{proof}
We first introduce some supporting notation. Let $M(p) = p^{(1)}, p^{(2)}, \ldots$ be the sequence of movements $p$ executes according to $A$. Let $M_i(p)$ denote the remaining sequence of movements in $M(p)$ after the forest schedule reaches $C_i$, and let $m_i(p)$ denote the first movement in $M_i(p)$.

\begin{claim}
A greedy forest schedule $\mathcal{S} = (A, (C_0, \ldots, C_t))$ can be constructed from configuration $C_0 = C_0^{(A)}$ such that, for every $0 \leq i \leq t$, configuration $C_i$ is obtained from $C_{i-1}$ by executing only the movements of a greedily selected, mutually compatible subset of $\{m_{i-1}(p) : p \in P\}$.
\end{claim}
\begin{proof}
Argue by induction on $i$, the current configuration number. $C_0$ is trivially obtained, as it is the initial configuration.
Assume by induction that the claim holds up to $C_{i-1}$. W.l.o.g.~let $M_{i-1} = \{m_{i-1}(p_1), \ldots, m_{i-1}(p_k)\}$, for $k \leq n$, be the greedily selected, mutually compatible subset of movements that $\mathcal{S}$ performs in moving from $C_{i-1}$ to $C_i$.
Suppose to the contrary that a movement $m'(p) \not\in M_{i-1}$ is executed by a particle $p \in P$. It is easily seen that $m'(p)$ cannot be $m_{i-1}(p)$; since $m_{i-1}(p)$ was excluded when $M_{i-1}$ was greedily selected, it must be incompatible with one or more of the selected movements and thus cannot also be executed at this time. So $m'(p) \neq m_{i-1}(p)$, and we have the cases below:

\begin{case}
$m_{i-1}(p)$ is a sole contraction. Then $p$ is expanded and has no children in $C_{i-1}$, so we must have $m'(p) = m_{i-1}(p)$, since there are no other movements $p$ could execute, a contradiction.
\end{case}

\begin{case}
$m_{i-1}(p)$ is a sole expansion. Then $p$ is contracted and has no parent in $C_{i-1}$, so we must have $m'(p) = m_{i-1}(p)$, since there are no other movements $p$ could execute, a contradiction.
\end{case}

\begin{case}
$m_{i-1}(p)$ is a handover contraction with $q$, one of its children. Then at some time in $\mathcal{S}$ before reaching $C_{i-1}$, $q$ became a descendant of $p$; thus, $q$ must also be a descendant of $p$ in $C_{i-1}$.
If $q$ is not a child of $p$ in $C_{i-1}$, there exists a particle $z \not\in \{p, q\}$ such that $q$ is a descendant of $z$, which is in turn a descendant of $p$. So in order for $m_{i-1}(p)$ to be a handover contraction with $q$, $M(z)$ must include actions which allow $z$ to ``bypass'' its ancestor $p$, which is impossible. So $q$ must be a child of $p$ in $C_{i-1}$, and must be contracted at the time $m_{i-1}(p)$ is performed.
If $q$ is also contracted in $C_{i-1}$, then once again we must have $m'(p) = m_{i-1}(p)$. Otherwise, $q$ is expanded in $C_{i-1}$, and must have become so before $C_{i-1}$ was reached. But this yields a contradiction: since $\mathcal{S}$ is greedy, $q$ would have contracted prior to this point by executing either a sole contraction if it has no children, or a handover contraction with a contracted child whose existence is guaranteed by Lemma~\ref{lem:expparentchild}, since every expanded parent in $C_0$ has a contracted child.
\end{case}

\begin{case}
$m_{i-1}(p)$ is a handover expansion with $q$, its unique parent. Then we must have that $m_{i-1}(q)$ is a handover contraction with $p$, and an argument analogous to that of Case 3 follows.
\end{case}
\vspace{-.27in}\qed
\end{proof}

We conclude by showing that each configuration of the greedy forest schedule $\mathcal{S}$ constructed according to the claim is dominated by its asynchronous counterpart.
Argue by induction on $i$, the configuration number. Since $C_0 = C_0^{(A)}$, we have that $C_0^{(A)} \succeq C_0$. Assume by induction that for all rounds $0 \leq r \leq i-1$, we have $C_r^{(A)} \succeq C_r$.
Consider any particle $p$. Since $\mathcal{S}$ is constructed using the exact set of movements $p$ executes according to $A$ and each time $p$ moves it decreases either its head distance or tail distance by $1$, it suffices to show that $p$ has performed at most as many movements in $\mathcal{S}$ up to $C_i$ as it has according to $A$ up to $C_i^{(A)}$.

If $p$ does not perform a movement between $C_{i-1}$ and $C_i$, we trivially have $C_i^{(A)}(p) \succeq C_i(p)$.
Otherwise, $p$ performs movement $m_{i-1}(p)$ to obtain $C_i$ from $C_{i-1}$. If $p$ has already performed $m_{i-1}(p)$ according to $A$ before reaching $C_{i-1}^{(A)}$, then clearly $C_i^{(A)}(p) \succeq C_i(p)$. Otherwise, $m_{i-1}(p)$ must be the next movement $p$ is to perform according to $A$, since $p$ has performed the same sequence of movements in the asynchronous execution as it has in $\mathcal{S}$ up to the respective rounds $i-1$, and thus has equal head and tail distances in $C_{i-1}$ and $C_{i-1}^{(A)}$.
It remains to show that $p$ can indeed perform $m_{i-1}(p)$ between $C_{i-1}^{(A)}$ and $C_i^{(A)}$.
If $m_{i-1}(p)$ is a sole expansion, then $p$ is the super-root of its tree (in both $C_{i-1}$ and $C_{i-1}^{(A)}$) and must also be able to expand in $C_{i-1}^{(A)}$. Similarly, if $m_{i-1}(p)$ is a sole contraction, then $p$ has no children (in both $C_{i-1}$ and $C_{i-1}^{(A)}$) and must be able to contract in $C_{i-1}^{(A)}$.
If $m_{i-1}(p)$ is a handover expansion with its parent $q$, then $q$ must be expanded in $C_{i-1}$. Parent $q$ must also be expanded in $C_{i-1}^{(A)}$; otherwise $d_h(q, C_{i-1}^{(A)}) > d_h(q, C_{i-1})$, contradicting the induction hypothesis. An analogous argument holds if $m_{i-1}(p)$ is a handover contraction with one of its contracted children. Therefore, in any case we have $C_i^{(A)}(p) \succeq C_i(p)$, and since the choice of $p$ was arbitrary, $C_i^{(A)} \succeq C_i$.
\qed
\end{proof}

We can show a similar dominance result when considering complaint flags.

\begin{definition}
\label{defn:complaintparallelschedule}
A movement schedule $(C_0, \ldots, C_t)$ is called a \emph{complaint-based parallel schedule} if each $C_i$ is a valid configuration of a particle system in which every particle holds at most \emph{one} complaint flag (rather than two, as described in Algorithm~\ref{alg:complaint}) and for every $i \geq 0$, $C_{i+1}$ is reached from $C_i$ such that for every particle $p$ one of the following properties holds:
\begin{enumerate}
\item $p$ does not hold a complaint flag and property 1, 3, or 4 of Definition~\ref{defn:parallelschedule} holds,
\item $p$ holds a complaint flag $f$ and expands into an adjacent node that was empty in $C_i$, consuming $f$,
\item $p$ forwards a complaint flag $f$ to a neighboring particle $p'$ which either does not hold a complaint flag in $C_i$ or is also forwarding its complaint flag.
\end{enumerate}
\end{definition}

A \emph{complaint-based forest schedule} $\mathcal{S} = (A, (C_0, \ldots, C_t))$ has the same properties as a forest schedule, with the exception that $(C_0, \ldots, C_t)$ is a complaint-based parallel schedule as opposed to a parallel schedule.
A complaint-based forest schedule is said to be \emph{greedy} if all particles perform movements according to Definition~\ref{defn:complaintparallelschedule} in the direction of their unique paths whenever possible.

We can now extend the dominance argument to hold with respect to \emph{complaint distance} in addition to head and tail distances. For any particle $p$ holding a complaint flag $f$ in configuration $C$, we define its complaint distance $d_c(f,C)$ to be the number of edges along $P_p$ from the node $p$ occupies to the end of $P_p$.
For any two configurations $C$ and $C'$ and any complaint flag $f$, we say that \emph{$C$ dominates $C'$ w.r.t.~$f$}, denoted $C(f) \succeq C'(f)$, if and only if $d_c(f,C) \leq d_c(f,C')$. Extending the previous notion of dominance, we say that \emph{$C$ dominates $C'$}, denoted $C \succeq C'$, if and only if $C$ dominates $C'$ with respect to every particle and with respect to every complaint flag.

It is also possible to construct a greedy complaint-based forest schedule whose configurations are dominated by their asynchronous counterparts, as we did for greedy forest schedules in Lemma~\ref{lem:forestdom}. Many of the details are the same, so as to avoid redundancy we highlight the differences here.
The most obvious difference is the inclusion of complaint flags. Definition~\ref{defn:complaintparallelschedule} restricts particles to holding at most one complaint flag at a time, where Algorithm~\ref{alg:complaint} allows a capacity of two. This allows the asynchronous execution to not ``fall behind'' the parallel schedule in terms of forwarding complaint flags. Basically, Definition~\ref{defn:complaintparallelschedule} allows a particle $p$ holding a complaint flag $f$ in the parallel schedule to forward $f$ to its parent $q$ even if $q$ currently holds its own complaint flag, so long as $q$ is also forwarding its flag at this time. The asynchronous execution does not have this luxury of synchronized actions, so the mechanism of buffering up to two complaint flags at a time allows it to ``mimic'' the pipelining of forwarding complaint flags that is possible within one round of a complaint-based parallel schedule.

Another slight difference is that a contracted particle cannot expand into an empty adjacent node unless it holds a complaint flag to consume. However, this restriction reflects Algorithm~\ref{alg:boundaryDirectionAlgorithm}, so once again the greedy complaint-based forest schedule can be constructed directly from the movements taken in the asynchronous execution. Moreover, since this restriction can only cause a contracted particle to remain contracted, the conditions of Lemma~\ref{lem:expparentchild} are still upheld.
Thus, we obtain the following lemma:

\begin{lemma}
\label{lem:flagforestdom}
Given any fair asynchronous activation sequence $A$ which begins at an initial configuration $C_0^{(A)}$ in which every expanded parent has at least one contracted child, there is a greedy complaint-based forest schedule $\mathcal{S} = (A,(C_0, \ldots, C_t))$ with $C_0 = C_0^{(A)}$ such that $C_i^{(A)} \succeq C_i$ for all $0 \leq i \leq t$.
\end{lemma}

By Lemmas~\ref{lem:forestdom} and~\ref{lem:flagforestdom}, once we have an upper bound for the time it takes a greedy forest schedule to reach a final configuration, we also have an upper bound for the number of rounds required by the asynchronous execution. Hence, the remainder of our proofs will serve to upper bound the number of parallel rounds any greedy forest schedule would require to solve the coating problem for a given valid instance $(P,O)$, where $|P| = n$. Let $\mathcal{S}^* = (A, (C_0, \ldots, C_f))$ be such a greedy forest schedule, where $C_0$ is the initial configuration of the particle system $P$ (of all contracted particles) and $C_f$ is the final coating configuration.

In Sections~\ref{sec:firstlayer} and~\ref{sec:higherlayers}, we will upper bound the number of parallel rounds required by $\mathcal{S}^*$ in the worst case to coat the first and higher layers, respectively. More specifically, we will bound the worst-case time it takes to complete a layer $i$ once layers $1, \ldots, i-1$ have been completed.
For convenience, we will not differentiate between complaint-based and regular forest schedules in the following sections, since the same dominance result holds whether or not complaint flags are considered.
To prove these bounds, we need one last definition: a \emph{forest--path schedule} $\mathcal{S} = (A, (C_0, \ldots, C_t), L)$ is a forest schedule $(A,(C_0, \ldots, C_t))$ with the property that all the trees of $A(C_0)$ are rooted at a path $L = v_1v_2\cdots v_\ell \subseteq \Geqt$, and each particle $p$ must traverse $L$ in the same direction.

\vspace{-.1in}
\subsection{First layer: complaint-based coating and leader election} \label{sec:firstlayer}
\vspace{-.05in}

Our algorithm must first organize the particles using the spanning forest primitive, whose runtime is easily bounded:

\begin{lemma}
Following the spanning forest primitive, the particles form a spanning forest within $\mathcal{O}(n)$ rounds.
\end{lemma}
\begin{proof}
Initially all particles are idle. In each round any idle particle adjacent to the object, an active (follower or root) particle, or a retired particle becomes active. It then sets its parent flag if it is a follower, or becomes the root of a tree if it is adjacent to the object or a retired particle. In each round at least one particle becomes active, so --- given $n$ particles in the system --- it will take $\mathcal{O}(n)$ rounds in the worst case until all particles join the spanning forest.
\qed
\end{proof}

For ease of presentation, we assume that the particle system is of sufficient size to fill the first layer (i.e., $B_1 \leq n$; the proofs can easily be extended to handle the case when $B_1 > n$); we also assume that the root of a tree also generates a complaint flag upon its activation (this assumption does not hurt our argument since it only increases the number of the flags generated in the system).
Let $\mathcal{S}_1 = (A, (C_0, \ldots, C_{t_1}), L_1)$ be the greedy forest--path schedule where $(A, (C_0, \ldots, C_{t_1}))$ is a truncated version of $\mathcal{S}^*$, $C_{t_1}$ --- for $t_1 \leq f$ --- is the configuration in $\mathcal{S}^*$ in which layer $1$ becomes complete, and $L_1$ is the path of nodes in layer $1$.
The following lemma shows that the algorithm makes steady progress towards completing layer $1$.

\begin{lemma}
\label{lem:progresslayer1}
Consider a round $i$ of the greedy forest--path schedule $\mathcal{S}_1$, where $0 \leq i \leq t_1-2$. Then within the next two parallel rounds of $\mathcal{S}_1$, $(i)$ at least one complaint flag is consumed, $(ii)$ at least one more complaint flag reaches a particle in layer $1$, $(iii)$ all remaining complaint flags move one position closer to a super-root along $L_1$, or $(iv)$ layer $1$ is completely filled (possibly with some expanded particles).
\end{lemma}
\begin{proof}
If layer 1 is filled, $(iv)$ is satisfied; otherwise, there exists at least one super-root in $A(C_i)$. We consider several cases:

\begin{case}
There exists a super-root $s$ in $A(C_i)$ which holds a complaint flag. If $s$ is contracted, then it can expand and consume its flag by the next round. Otherwise, consider the case when $s$ is expanded. If it has no children, then within the next two rounds it can contract and expand again, consuming its complaint flag; otherwise, by Lemma~\ref{lem:expparentchild}, $s$ must have a contracted child with which it can perform a handover to become contracted in $C_{i+1}$ and then expand and consume its complaint flag by $C_{i+2}$. In any case, $(i)$ is satisfied.
\end{case}

\begin{case}
No super-root in $A(C_i)$ holds a complaint flag and not all complaint flags have been moved from follower particles to particles in layer 1. Let $p_1, p_2, \ldots, p_z$ be a sequence of particles in layer 1 such that each particle holds a complaint flag, no follower child of any particle except $p_z$ holds a complaint flag, and no particles between the next super-root $s$ and $p_1$ hold complaint flags. Then, as each $p_i$ forwards its flag to $p_{i-1}$ according to Definition~\ref{defn:complaintparallelschedule}, the follower child of $p_z$ holding a flag is able to forward its flag to $p_z$, satisfying $(ii)$.
\end{case}

\begin{case}
No super-root in $A(C_i)$ holds a complaint flag and all remaining complaint flags are held by particles in layer 1. By Definition~\ref{defn:complaintparallelschedule}, since no preference needs to be given to flags entering layer 1, all remaining flags will move one position closer to a super-root in each round, satisfying $(iii)$.
\end{case}
\vspace{-0.27in}\qed
\end{proof}

We use Lemma~\ref{lem:progresslayer1} to show first that layer $1$ will be filled with particles (some possibly still expanded) in $\mathcal{O}(n)$ rounds. From that point on, in another $\mathcal{O}(n)$ rounds, one can guarantee that expanded particles in layer $1$ will each contract in a handover with a follower particle, and hence all particles in layer $1$  will be contracted, as we see in the following lemma:

\begin{lemma}
\label{lemma:filled}
After $\mathcal{O}(n)$ rounds, layer 1 must be filled with contracted particles.
\end{lemma}
\begin{proof}
We first prove the following claim:

\begin{claim}
After $8B_1+2$ rounds of $\mathcal{S}$, layer $1$ must be filled with particles.
\end{claim}
\begin{proof}
Suppose to the contrary that after $8B_1 + 2$ rounds, layer 1 is not completely filled with particles. Then none of these rounds could have satisfied $(iv)$ of Lemma~\ref{lem:progresslayer1}, so one of $(i), (ii)$, or $(iii)$ must be satisfied every two rounds.
Case $(i)$ can be satisfied at most $B_1$ times (accounting for at most $2B_1$ rounds), since a super-root expands into an unoccupied position of layer 1 each time a complaint flag is consumed.
Case $(iii)$ can also be satisfied at most $B_1$ times (accounting for at most $2B_1$ rounds), since once all remaining complaint flags are in layer 1, every flag must reach a super-root in $B_1$ moves.
Thus, the remaining $8B_1+2 - 2B_1 - 2B_1 = 4B_2+2$ rounds must satisfy $(ii)$ $2B_1+1$ times, implying that $2B_1+1$ flags reached particles in layer 1 from follower children. But each particle can hold at most one complaint flag, so at least $B_1+1$ flags must have been consumed and the super-roots have collectively expanded into at least $B_1+1$ unoccupied positions, a contradiction.
\qed
\end{proof}

By the claim, it will take at most $8B_1+2$ rounds until layer $1$ is completely filled with particles (some possibly expanded). In at most another $B_1$ rounds, every expanded particle in layer $1$ will contract in a handover with a follower particle (since $B_1\leq n$), and hence all particles in layer $1$ will be contracted after $\mathcal{O}(B_1) = \mathcal{O}(n)$ rounds.
\qed
\end{proof}

Once layer $1$ is filled, the leader election primitive can proceed. The full description of the Universal Coating algorithm in~\cite{Derakhshandeh2017} uses a node-based version of the leader election algorithm in~\cite{Daymude2016} for this primitive. For consistency, we kept this description of the primitive in this paper as well. However, with high probability guarantees were not proven for the leader election algorithm in~\cite{Daymude2016}. In order to formally prove with high probability results on the runtime of the universal coating algorithm, we introduced a variant of our leader election protocol under the amoebot model which provably elects a leader in $\mathcal{O}(n)$ asynchronous rounds, w.h.p.\footnote{The updated leader election algorithm's runtime holds with high probability, but its correctness is guaranteed; see~\cite{Daymude2017} for details.}~\cite{Daymude2017}. This gives the following runtime bound.

\begin{lemma}
\label{lem:leaderelection}
A position of layer 1 will be elected as the leader position, and w.h.p.~this will occur within $\mathcal{O}(n)$ additional rounds.
\end{lemma}

Once a leader position has been elected and either no more followers exist (if $n \leq B_1$) or all positions are completely filled by contracted particles (which can be checked in an additional $\mathcal{O}(B_1)$ rounds), the particle currently occupying the leader position becomes the leader particle. Once a leader has emerged, the particles on layer $1$ retire, which takes $\mathcal{O}(B_1)$ further rounds. Together, we get:

\begin{corollary}
\label{cor:layer1}
The worst-case number of rounds for $\mathcal{S}^*$ to complete layer $1$ is $\mathcal{O}(n)$, w.h.p.
\end{corollary}

\vspace{-.1in}
\subsection{Higher layers} \label{sec:higherlayers}
\vspace{-.05in}

We again use the dominance results we proved in Section~\ref{sec:parallel} to focus on parallel schedules when proving an upper bound on the worst-case number of rounds --- denoted by $Layer(i)$ --- for building layer $i$ once layer $i-1$ is complete, for $2 \leq i \leq N$. The following lemma provides a more general result which we can use for this purpose.

\begin{lemma}
\label{lem:tree-time}
Consider any greedy forest--path schedule $\mathcal{S} = (A, (C_0, \ldots, C_t), L)$ with $L = v_1v_2\cdots v_\ell$ and any $k$ such that $1 \leq k \leq \ell$. If every expanded parent in $C_0$ has at least one contracted child, then in at most $2(\ell + k)$ configurations, nodes $v_{\ell -k+1}\cdots v_\ell$ will be occupied by contracted particles.
\end{lemma}
\begin{proof}
Let $s$ be the super-root closest to $v_\ell$, and suppose $s$ initially occupies node $v_i$ in $C_0$. Additionally, suppose there are at least $k$ active particles in $C_0$ (otherwise, we do not have sufficient particles to occupy $k$ nodes of $L$).
Argue by induction on $k$, the number of nodes in $L$ starting with $v_\ell$ which must be occupied by contracted particles.
First suppose that $k = 1$. By Lemma~\ref{lem:expparentchild}, every expanded parent has at least one contracted child in any configuration $C_j$, so $s$ is always able to either expand forward into an unoccupied node of $L$ if it is contracted or contract as part of a handover with one of its children if it is expanded. Thus, in at most $2(\ell + k) = 2\ell + 2$ configurations, $s$ has moved forward $\ell$ positions, is contracted, and occupies its final position $v_{\ell - k + 1} = v_\ell$.

Now suppose that $k > 1$ and that each node $v_{\ell - x + 1}$, for $1 \leq x \leq k-1$, becomes occupied by a contracted particle in at most $2(\ell + k - 1) = 2(\ell + k) - 2$ configurations. It suffices to show that $v_{\ell - k + 1}$ also becomes occupied by a contracted particle in at most two additional configurations. Let $p$ be the particle currently occupying $v_{\ell - k + 1}$ (such a particle must exist since we supposed we had sufficient particles to occupy $k$ nodes and $\mathcal{S}$ ensures the particles follow this unique path).
If $p$ is contracted in $C_{2(\ell + k) - 2}$, then it remains contracted and occupying $v_{\ell - k + 1}$, so we are done.
Otherwise, if $p$ is expanded, it has a contracted child $q$ by Lemma~\ref{lem:expparentchild}. Particles $p$ and $q$ thus perform a handover in which $p$ contracts to occupy only $v_{\ell - k + 1}$ at $C_{2(\ell + k) - 1}$, proving the claim.
\qed
\end{proof}

For convenience, we introduce some additional notation. Let $n_i$ denote the number of particles of the system that will not belong to layers 1 through $i-1$, i.e., $n_i=n-\sum_{j=1}^{i-1}B_j$, and let $t_i$ (resp., $C_{t_i}$) be the round (resp., configuration) in which layer $i$ becomes complete.

When coating some layer $i$, each root particle either moves either $(a)$ through the nodes in layer $i$ in the set direction $dir$ (CW or CCW) for layer $i$, or $(b)$ through the nodes in layer $i+1$ in the opposite direction over the already retired particles in layer $i$ until it finds an empty position in layer $i$. We bound the worst-case scenario for these two movements independently in order to get a an upper bound on $Layer(i)$. Let $L_i = v_1, \ldots, v_{B_i}$ be the path of nodes in layer $i$ listed in the order that they appear from the marker position $v_1$ following direction $dir$, and let $\mathcal{S}_i = (A, (C_{t_{i-1}+1}, \ldots, C_{t_i}), L_i)$ be the greedy forest--path schedule where $(A, (C_{t_{i-1}+1}, \ldots, C_{t_i}))$ is a section of $\mathcal{S}^*$. By Lemma~\ref{lem:tree-time}, it would take $\mathcal{O}(B_i)$ rounds for all $(a)$ movements to complete; an analogous argument shows that all $(b)$ movements complete in $\mathcal{O}(B_{i+1}) = \mathcal{O}(B_i)$ rounds.
This implies the following lemma:

\begin{lemma}
\label{lem:layeri2}
Starting from configuration $C_{t_{i-1} + 1}$, the worst-case additional number of rounds for layer $i$ to become complete is $\mathcal{O}(B_i)$.
\end{lemma}

Putting it all together, for layers $2$ through $N$:

\begin{corollary}
\label{cor:higherlayers}
The worst-case number of rounds for $\mathcal{S}^*$ to coat layers 2 through $N$ is $\mathcal{O}(n)$.
\end{corollary}
\begin{proof}
Starting from configuration $C_{t_1 + 1}$, it follows from Lemma~\ref{lem:layeri2} that the worst-case number of rounds for $\mathcal{S}^*$ to reach a legal coating of the object is upper bounded by
\[\sum_{i=2}^{N}Layer(i) \leq c\sum_{i=2}^{N}B_i = \Theta(n),\]
where $c>0$ is a constant.
\qed
\end{proof}

Combining Corollaries~\ref{cor:layer1} and~\ref{cor:higherlayers}, we get that $\mathcal{S}^*$ requires $\mathcal{O}(n)$ rounds w.h.p.~to coat any given valid object $O$ starting from any valid initial configuration of the set of particles $P$. By Lemmas~\ref{lem:forestdom} and~\ref{lem:flagforestdom}, the worst-case behavior of $\mathcal{S}^*$ is an upper bound for the runtime of our Universal Coating algorithm, so we conclude:

\begin{theorem}
\label{thm:chain}
The total number of asynchronous rounds required for the Universal Coating algorithm to reach a legal coating configuration, starting from an arbitrary valid instance $(P,O)$, is $\mathcal{O}(n)$ w.h.p., where $n$ is the number of particles in the system.
\end{theorem}

\vspace{-.1in}
\section{Simulation Results} \label{sec:experimental}
\vspace{-.05in}
In this section we present a brief simulation-based analysis of our algorithm which shows that in practice our algorithm exhibits a better than linear average competitive ratio.
Since $\OPT(P,O)$ (as defined in Section~\ref{sec:performance}) is difficult to compute in general, we investigate the competitiveness with the help of an appropriate lower bound for $\OPT(P,O)$. Recall the definitions of the distances $d(p,q)$ and $d(p,U)$ for $p, q \in \Veqt$ and $U \subseteq \Veqt$. Consider any valid instance $(P,O)$. Let $\mathcal{L}$ be the set of all legal particle positions of $(P,O)$; that is, $\mathcal{L}$ contains all sets $U \subseteq \Veqt$ such that the positions in $U$ constitute a coating of the object $O$ by the particles in the system.

We compute a lower bound on $\OPT(P,O)$ as follows.
Consider any $U \in \mathcal{L}$, and let $G(P,U)$ denote the complete bipartite graph on partitions $P$ and $U$. For each edge $e = (p, u) \in P \times U$, set the cost of the edge to $w(e) = d(p,u)$. Every perfect matching in $G(P,U)$ corresponds to an assignment of the particles to positions in the coating. The maximum edge weight in a matching corresponds to the maximum distance a particle has to travel in order to take its place in the coating.
Let $M(P,U)$ be the set of all perfect matchings in $G(P,U)$. We define the \emph{matching dilation} of $(P,O)$ as
\[\text{MD}(P,O) = \min_{U \in \mathcal{L}} \left( \min_{M \in M(P,U)} \left( \max_{e \in M} \left( w(e) \right)\right)\right).\]
Since each particle has to move to some position in $U$ for some $U \in \mathcal{L}$ to solve the coating problem, we have $\OPT(P,O) \geq \text{MD}(P,O)$. The search for the matching that minimizes the maximum edge cost for a given $U \in \mathcal{L}$ can be realized efficiently by reducing it to a flow problem using edges up to a maximum cost of $c$ and performing binary search on $c$ to find the minimal $c$ such that a perfect matching exists.
We note that our lower bound is not tight. This is due to the fact that it only respects the distances that particles have to move but ignores the congestion that may arise, i.e., in certain instances the distances to the object might be very small, but all particles may have to traverse one ``chokepoint'' and thus block each other.

\begin{figure}
    \centering
    \includegraphics[width = \textwidth]{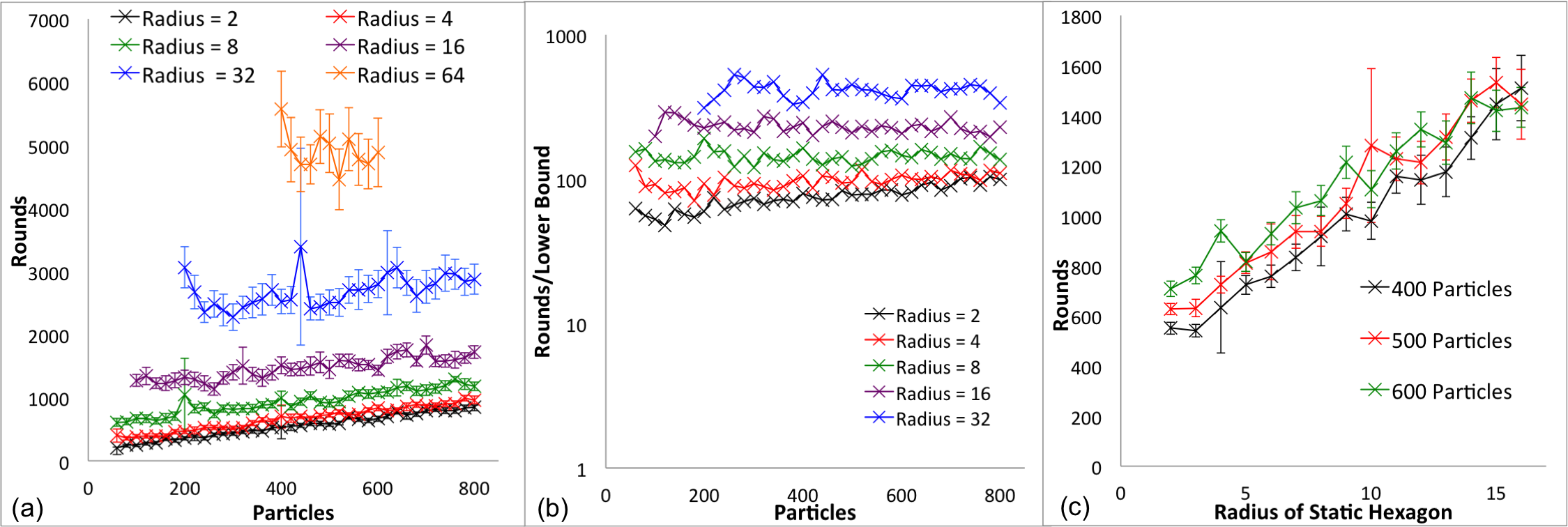}
    \caption[Simulation Results]{\small (a) shows the number of rounds varying the number of particles. (b) shows the ratio of number of rounds to the lower bound in log scale. (c) shows the number of rounds varying the static hexagon radius.}
    \vspace{-.15in}
    \label{fig:simresults}
\end{figure}

We implemented the Universal Coating algorithm in the amoebot simulator (see \cite{sops-webpage} for videos). For simplicity, each simulation is initialized with the object $O$ as a regular hexagon of object particles; this is reasonable since the particles need only know where their immediate neighbors in the object's border are relative to themselves, which can be determined independently of the shape of the border. The particle system $P$ is initialized as idle particles attached randomly around the hexagon's perimeter. The parameters that were varied between instances are the radius of the hexagon and the number of (initially idle) particles in $P$. Each experimental trial randomly generates a new initial configuration of the system.

Figure~\ref{fig:simresults}(a) shows the number of rounds needed to complete the coating with respect to the hexagon object radius and the number of particles in the system. The number of rounds plotted are averages over 20 instances of a given $|P|$ with 95\% confidence intervals. These results show that, in practice, the number of rounds required increases linearly with particle system size. This agrees with our expectations, since leader election depends only on the length of the object's surface while layering depends on the total number of particles.
Figure~\ref{fig:simresults}(b) shows the ratio of the number of rounds to the matching dilation of the system. These results indicate that, in experiment, the average competitive ratio of our algorithm may exhibit closer to logarithmic behaviors.
Figure~\ref{fig:simresults}(c) shows the number of rounds needed to complete the coating as the radius of the hexagon object is varied. The runtime of the algorithm appears to increase linearly with both the number of active particles and the size of the object being coated, and there is visibly increased runtime variability for systems with larger radii.

\vspace{-.1in}
\section{Conclusion} \label{sec:concl}
\vspace{-.05in}
This paper continued the study of universal coating in self-organizing particle systems. The runtime analysis shows that our Universal Coating algorithm terminates in a linear number of rounds with high probability, and thus is worst-case optimal. This, along with the linear lower bound on the competitive gap between local and global algorithms, further shows our algorithm to be competitively optimal. Furthermore, the simulation results indicate that the competitive ratio of our algorithm may be better than linear in practice. In the future, we would like to apply the algorithm and analysis to the case of bridging, in which particles create structures across gaps between disconnected objects.
We would also like to extend the algorithm to have self-stabilization capabilities, so that it could successfully complete coating without human intervention after occasional particle failures or outside interference.

\bibliographystyle{unsrt}
{\small \bibliography{literature}}

\begin{thebibliography}{10}

\bibitem{Daymude2016}
J.~J. Daymude, Z.~Derakhshandeh, R.~Gmyr, T.~Strothmann, R.~A. Bazzi, A.~W.
  Richa, and C.~Scheideler.
\newblock Leader election and shape formation with self-organizing programmable
  matter.
\newblock {\em CoRR}, abs/1503.07991, 2016.
\newblock A preliminary version of this work appeared in {DNA}21, 2015, pp.
  117--132.

\bibitem{Derakhshandeh2014}
Z.~Derakhshandeh, S.~Dolev, R.~Gmyr, A.~W. Richa, C.~Scheideler, and
  T.~Strothmann.
\newblock Brief announcement: amoebot - a new model for programmable matter.
\newblock In {\em Proc. of the 26th {ACM} Symposium on Parallelism in
  Algorithms and Architectures ({SPAA} '14)}, pages 220--222, 2014.

\bibitem{Derakhshandeh2017}
Z.~Derakhshandeh, R.~Gmyr, A.~W. Richa, C.~Scheideler, and T.~Strothmann.
\newblock Universal coating for programmable matter.
\newblock {\em Theoretical Computer Science}, 671:56--68, 2017.

\bibitem{Lynch1996}
N.~Lynch.
\newblock {\em Distributed Algorithms}.
\newblock Morgan Kauffman, 1996.

\bibitem{Angluin2006}
D.~Angluin, J.~Aspnes, Z.~Diamadi, M.~J. Fischer, and R.~Peralta.
\newblock Computation in networks of passively mobile finite-state sensors.
\newblock {\em Distributed Computing}, 18(4):235--253, 2006.

\bibitem{Doty2012}
D.~Doty.
\newblock Theory of algorithmic self-assembly.
\newblock {\em Communications of the {ACM}}, 55(12):78--88, 2012.

\bibitem{Patitz2014}
M.~J. Patitz.
\newblock An introduction to tile-based self-assembly and a survey of recent
  results.
\newblock {\em Natural Computing}, 13(2):195--224, 2014.

\bibitem{Woods2015}
D.~Woods.
\newblock Intrinsic universality and the computational power of self-assembly.
\newblock {\em Philosophical Transactions of the Royal Society A}, 373(2046),
  2015.

\bibitem{Bonifaci2012}
V.~Bonifaci, K.~Mehlhorn, and G.~Varma.
\newblock Physarum can compute shortest paths.
\newblock {\em Journal of Theoretical Biology}, 309:121--133, 2012.

\bibitem{Li2010}
K.~Li, K.~Thomas, C.~Torres, L.~Rossi, and C.-C. Shen.
\newblock Slime mold inspired path formation protocol for wireless sensor
  networks.
\newblock In {\em Proc. of the 7th Int. Conference on Swarm Intelligence
  ({ANTS} '10)}, pages 299--311, 2010.

\bibitem{Wilson2014}
S.~Wilson, T.~Pavlic, G.~Kumar, A.~Buffin, S.~C. Pratt, and S.~Berman.
\newblock Design of ant-inspired stochastic control policies for collective
  transport by robotic swarms.
\newblock {\em Swarm Intelligence}, 8(4):303--327, 2014.

\bibitem{Brambilla2013}
M.~Brambilla, E.~Ferrante, M.~Birattari, and M.~Dorigo.
\newblock Swarm robotics: a review from the swarm engineering perspective.
\newblock {\em Swarm Intelligence}, 7(1):1--41, 2013.

\bibitem{Kumar2014}
G.~P. Kumar and S.~Berman.
\newblock Statistical analysis of stochastic multi-robot boundary coverage.
\newblock In {\em Proc. of the 2014 {IEEE} Int. Conference on Robotics and
  Automation ({ICRA} '14)}, pages 74--81, 2014.

\bibitem{Pavlic2016}
T.~Pavlic, S.~Wilson, G.~Kumar, and S.~Berman.
\newblock {\em An Enzyme-Inspired Approach to Stochastic Allocation of Robotic
  Swarms Around Boundaries}, pages 631--647.
\newblock Springer, 2016.

\bibitem{Blazovics2012-CFC}
L.~Bl{\'{a}}zovics, K.~Csorba, B.~Forstner, and H.~Charaf.
\newblock Target tracking and surrounding with swarm robots.
\newblock In {\em Proc. of the 19th {IEEE} Int. Conference and Workshops on the
  Engineering of Computer Based Systems ({ECBS} '12)}, pages 135--141, 2012.

\bibitem{Blazovics2012-LF}
L.~Bl{\'{a}}zovics, T.~Lukovszki, and B.~Forstner.
\newblock Target surrounding solution for swarm robots.
\newblock In {\em Information and Communication Technologies ({EUNICE} '12)},
  pages 251--262, 2012.

\bibitem{Michail2016}
O.~Michail and P.~G. Spirakis.
\newblock Simple and efficient local codes for distributed stable network
  construction.
\newblock {\em Distributed Computing}, 29(3):207--237, 2016.

\bibitem{Woods2013}
D.~Woods, H.-L. Chen, S.~Goodfriend, N.~Dabby, E.~Winfree, and P.~Yin.
\newblock Active self-assembly of algorithmic shapes and patterns in
  polylogarithmic time.
\newblock In {\em Proc. of the 4th Conference on Innovations in Theoretical
  Computer Science ({ITCS} '13)}, pages 353--354, 2013.

\bibitem{Chen2015}
M.~Chen, D.~Xin, and D.~Woods.
\newblock Parallel computation using active self-assembly.
\newblock {\em Natural Computing}, 14(2):225--250, 2015.

\bibitem{Daymude2017}
J.~J. Daymude, R.~Gmyr, A.~W. Richa, C.~Scheideler, and T.~Strothmann.
\newblock Improved leader election for self-organizing programmable matter.
\newblock {\em CoRR}, abs/1701.03616, 2017.
\newblock Submitted to {ALGOSENSORS} '17.

\bibitem{Derakhshandeh2015-nanocom}
Z.~Derakhshandeh, R.~Gmyr, A.~W. Richa, C.~Scheideler, and T.~Strothmann.
\newblock An algorithmic framework for shape formation problems in
  self-organizing particle systems.
\newblock In {\em Proc. of the 2nd Int. Conference on Nanoscale Computing and
  Communication ({NanoCom} '15)}, pages 21:1--21:2, 2015.

\bibitem{sops-webpage}
Self-organizing particle systems.
\newblock sops.engineering.asu.edu/simulations/.

\end{thebibliography}

\end{document}